\documentclass[a4paper,english,numberwithinsect]{eurocg25}

\title{A Linear Time Algorithm for Finding Minimum Flip Sequences between Plane Spanning Paths in Convex Point Sets\footnote{This research was funded in part by the Austrian Science
		Fund (FWF) 10.55776/DOC183.}}
\titlerunning{A Linear Time Algorithm for Minimum Flip Sequences}

\author[1]{Oswin Aichholzer}
\author[2]{Joseph Dorfer}
\affil[1]{Graz University of Technology\\
	\texttt{oswin.aichholzer@tugraz.at}}
\affil[2]{Graz University of Technology\\
	\texttt{joseph.dorfer@tugraz.at}}
\authorrunning{O. Aichholzer and J. Dorfer}

\graphicspath{{./figures/}}

\theoremstyle{plain}

\theoremstyle{definition}

\theoremstyle{remark}

\theoremstyle{definition}

\begin{document}
	
		\maketitle
	
	\begin{abstract}
		We provide a linear time algorithm to determine the flip distance between two plane spanning paths on a point set in convex position. At the same time, we show that the happy edge property does not hold in this setting. This has to be seen in contrast to several results for reconfiguration problems where the absence of the happy edge property implies algorithmic hardness of the flip distance problem. Further, we show that our algorithm can be adapted for (1) compatible flips (2) local flips and (3) flips for plane spanning paths in simple polygons.
	\end{abstract}
	
	\section{Introduction}
	
	Let $S$ be a finite point set in the plane in general position, that is, no three points lie on a common line. We call $S$ a \emph{convex} point set if no point of $S$ lies in the interior of the convex hull of $S$. For a convex point set with $n$ vertices we label the points $v_0$,\ldots,$v_{n-1}$ in clockwise order, starting with an arbitrary vertex. A plane straight-line drawing of a graph on $S$ is a graph with vertex set $S$ and whose edges are straight line segments between pairs of points of $S$ such that no two edges intersect, except at a common endpoint. Throughout this paper we will refer to graphs and their drawings interchangeably.
	
	\textbf{Flips in Plane Spanning Paths.} A \emph{plane spanning path} (or in the following sometimes just \emph{path}) on a point set $S$ is a plane graph with $\lvert S \rvert -1$ edges that is cycle-free, and in which every vertex has degree at most $2$. A \emph{flip} is an operation that removes one edge from a plane spanning path and adds another edge such that the resulting structure is again a plane spanning path. With this notion of a flip we can define the \emph{flip graph} as an abstract graph that has as vertex set the set of all plane spanning paths on $S$. Two vertices of the flip graph are connected by an edge if and only if their corresponding paths can be transformed into one another via a single flip.
	Given an initial plane spanning path $P_{in}$ and a target path~$P_{tar}$ on S, a \emph{flip sequence} from $P_{in}$ to $P_{tar}$ is a sequence of plane spanning paths $P_{in}=P_0$, $P_1$, \ldots, $P_k=P_{tar}$ such that two consecutive paths differ only by a flip. Equivalently, a flip sequence can be described as a path from $P_{in}$ to $P_{tar}$ in the flip graph.
	The \emph{length} of the flip sequence is the number $k$ of flips to transform $P_{in}$ into $P_{tar}$. The \emph{flip distance} between~$P_{in}$ and~$P_{tar}$ is the minimum $k$ for which a flip sequence from $P_{in}$ to $P_{tar}$ of length $k$ exists. A flip sequence that realizes this minimum will be called a \emph{minimum flip sequence}. Again, in the notion of a flip graph, minimum flip sequences correspond to shortest paths in the flip graph and the flip distance denotes the length of such a shortest path.
	
	In \cite{aichholzer2023flipping} a characterization of flips in plane spanning paths into three types is given. See Figure \ref{types} for an illustration. Consider a path $P~=~p_1,p_2,\ldots,p_n$ where the points are given in order of traversal by $P$. A flip of \emph{Type~1} removes an edge $(p_{i-1},p_i)$ and adds an edge $(p_1,p_i)$ or $(p_{i-1},p_n)$ if they do not cross any other edge. This results in a new path,~$p_{i-1},\ldots,p_1,p_i,\ldots,p_n$ or $p_1,\ldots,p_{i-1},p_n,\ldots,p_{i}$, respectively. A flip of \emph{Type~2} adds the edge~$(p_1,p_n)$ assuming this edge does not cross any of the already existing edges. Afterwards, an arbitrary edge~$(p_{i-1},p_i)$ from the original path can be removed. Note that adding and removing consecutive edges when closing a cycle can be interpreted as both a Type~1 flip and a Type~2 flip. For simplicity of notation we will count such flips as Type~2 flips. A flip of \emph{Type~3} also adds the edge~$(p_1,p_n)$ but now it is assumed to intersect exactly one edge~$(p_{i-1},p_i)$ in~$P$ which is then removed within this flip operation.
	
	\begin{figure}[ht]
		\centering
		\includegraphics[width=0.7\textwidth]{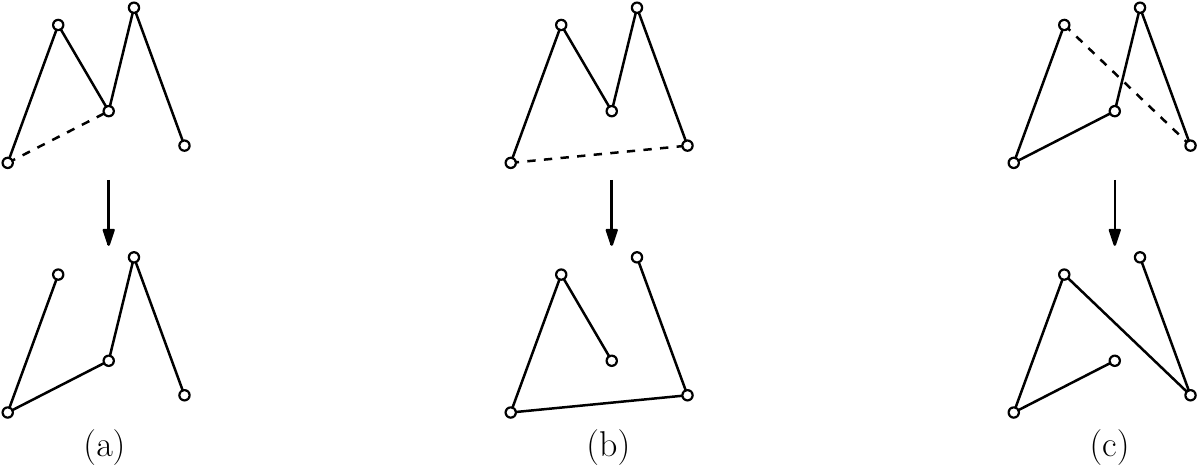}
		\caption{Flips in plane spanning paths: (a) Type 1 flip (b) Type 2 flip (c) Type 3 flip}
		\label{types}
	\end{figure}
	
	\textbf{Related work.} There are three central questions about reconfiguration problems that are considered in the literature:
	
	\begin{itemize}
		\item[(1)] Is the flip graph connected, that is, can we transform any configuration into any other via our given flip operation?
		\item[(2)] What is the diameter of the flip graph, that is, how many flips does it take in the worst case to flip one configuration into another configuration?
		\item[(3)] What is the complexity of computing minimum flip sequences between a given pair of configurations?
	\end{itemize}
	
	It is an interesting open problem, whether every plane spanning path of a given point set can be transformed into any other spanning path on the same point set. For the special cases of convex point sets \cite{AKL200759}, wheel sets, generalized double circles \cite{aichholzer2023flipping}, and point sets with at most two convex layers~\cite{kleist2024connectivityflipgraphplane}, it has been shown that we can always flip from one path to another. In \cite{kleist2024connectivityflipgraphplane} the existence of a large connected component in the flip graph of plane spanning paths-called suffix-independent paths-is derived. Notably, all the currently known connectivity results do not rely on any Type 3 flips. In the case of convex point sets the diameter of the flip graph is known to be exactly $2n-6$ for~$n>4$~\cite{CHANG2009409}.
	
	A related flip operation is given by flips in plane spanning trees. For plane spanning trees the flip graph is known to be connected with radius $n-2$~\cite{AVIS199621,bjerkevik2024flippingnoncrossingspanningtrees}. There have been several results with regard to bounds on the diameter of the flip graph~\cite{aichholzer2024reconfiguration,bjerkevik2024flippingnoncrossingspanningtrees,bousquet_et_al:LIPIcs.SoCG.2024.22,bousquet2023notes,hhmmn-gtgpcp-99}. The currently best known bounds of the diameter are $\frac{14n}{9}-O(1)$ as a lower bound and~$\frac{5n}{3}-3$ as an upper bound~\cite{bjerkevik2024flippingnoncrossingspanningtrees}. For more related results on flips in other planar graph structures we refer to the survey by Bose and Hurtado \cite{BOSE200960}. 
	
	\textbf{Happy Edges.} \emph{Happy edges} are edges that lie in both, the initial configuration and the target configuration of a graph reconfiguration problem. The so-called \emph{happy edge property} says that there always exists a minimum flip sequence between two configurations in which no edge is removed and later flipped in again. This implies that a happy edge is never flipped in such a sequence. The happy edge property may or may not hold for certain reconfiguration problems. Whether the happy edge property holds, can be a good indication for the complexity of a reconfiguration problem. For example, for triangulations of simple polygons~\cite{Aichholzer2015} and general point sets~\cite{LUBIW201517,Pilz_2014} finding minimum flip sequences is NP-hard and the gadgets in the proofs are built around conterexamples to the happy edge property. Conversely, the happy edge property holds for plane perfect matchings of convex point sets and a minimum flip sequence can be found in polynomial time \cite{articlematchings}. On the other hand, the happy edge property holds for triangulations of convex polygons~\cite{2ae353801697435f901a750248959ba2}, but the question about the complexity of finding minimum flip sequences is still open.
	
	
	\textbf{Our Contributions.} In this paper we show that the happy edge property does not hold for plane spanning paths of convex point sets. This adds to the already known counterexamples in point sets in general position \cite{Felsner}. Simultaneously, we provide an approach for finding minimum flip sequences between pairs of plane spanning paths on convex point sets in linear time. Interestingly, this is in contrast to all the previously mentioned results where the absence of the happy edge property implied hardness of finding minimum flip sequences.
	
	Further we show that our observations carry over to some variants of the problem. Our characterizations also yield a linear time algorithm if we (1) restrict the flips to be compatible, that is, we do not allow crossing Type 3 flips to happen; (2) make the problem more local by only allowing flips in which the removed edge and the added edge share an endpoint. This can also be interpreted as assigning different costs to Type 2 flips based on how non-local they are; or (3) consider paths in simple polygons.

	\section{Counterexample to the Happy Edge Property}
	
	\begin{figure}[ht]
		\centering
		\includegraphics[width=0.9\textwidth]{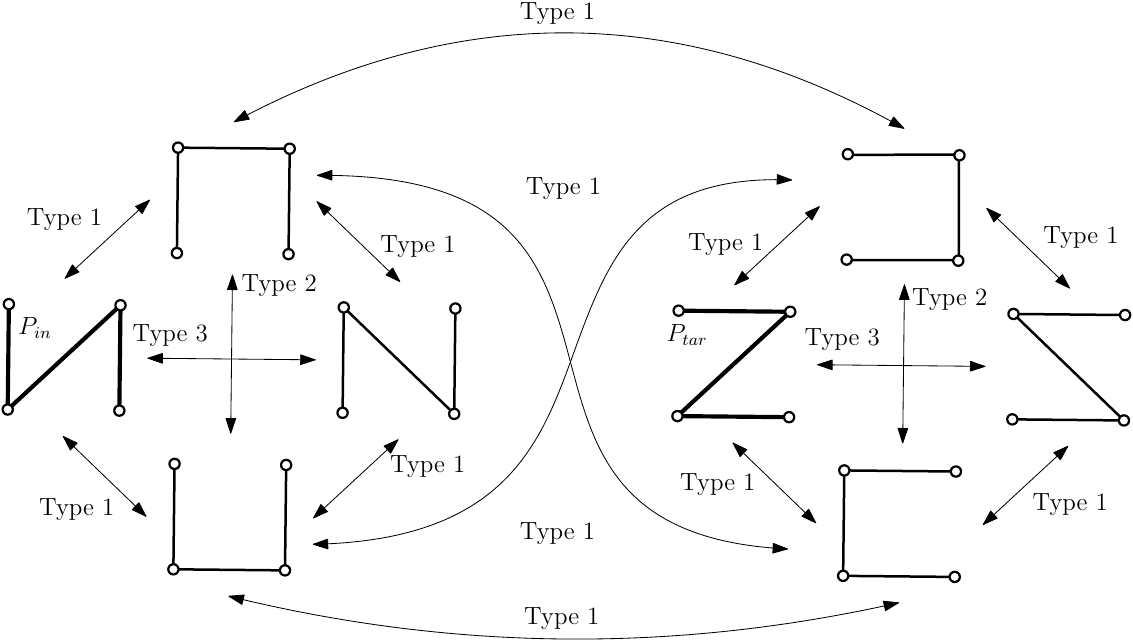}
		\caption[Smallest Counterexample to the happy edge property for spanning paths]{Flip graph of paths on four points in convex position with a counterexample to the happy edge property. Indicated are the different types of flips and the initial and target paths.}
		\label{smaller_example}
	\end{figure}
	
	\break
	
	In Figure~\ref{smaller_example} we show the flip graph of all eight plane spanning paths on four points in convex position. The initial path $P_{in}$ and the target path $P_{tar}$ are marked. Both paths share a diagonal, but no other path contains that particular diagonal. Moreover, the two paths cannot be directly transformed into one another via one flip. Therefore, this pair of paths provides a counterexample to the happy edge property for minimum flip sequences between plane spanning paths in convex point sets. Actually, in this example there is no flip sequence from $P_{in}$ to $P_{tar}$ at all that does not need to flip the happy edge.
	
	\section{Basics on Flips in Paths\label{basics}}
	
	In this section we show some basic relations for flips between plane paths.
	Due to visibility constraints Type 2 flips can only happen when all edges of the current path lie on the convex hull. 
	For Type 1 flips we obtain the following result.
	
	\begin{lemma} \label{obs1}
		Type 1 flips either remove a diagonal or add a diagonal, but not both at the same time. If diagonals exists, we can always lower the number of diagonals by one Type 1 flip.
	\end{lemma}
	
	\begin{proof}
		W.l.o.g.\ we say that we flip from a path $P_{in}$ to a path $P_{tar}$ where $v$ is the end vertex of $P_{in}$ that is involved in the Type 1 flip. See Figure \ref{Case12} for an illustration.
		We distinguish between two cases, depending on which type of edge is added during the flip.
		
		\textbf{Case 1:} A convex hull edge is added. We know that not all edges can be located on the convex hull, otherwise the only flip that can add a convex hull edge is a Type 2 flip, which we excluded by assumption. Therefore, there exist diagonals in $P_{in}$. 
		The only way, how a diagonal $d=(v_a,v_b)$ can be visible from an end vertex $v$ of $P_{in}$ is, if $d$ is the first diagonal to appear when traversing $P_{in}$ starting from $v$. All edges of $P_{in}$ that appear before~$d$ when traversing $P_{in}$ are convex hull edges.
		
		Without loss of generality, assume $v_a$ is reached before $v_b$ when traversing $P_{in}$ starting from~$v$. The only convex hull edge that contains $v$ which can be added to $P_{in}$ is $(v,v_b)$. Adding this edge results in removing the diagonal~$d$ since it is the unique next edge along the cycle that is closed within $P_{in}$ when adding $(v,v_b)$.	
		Therefore, $P_{tar}$ contains one diagonal less than $P_{in}$.
		
		\textbf{Case 2:} A diagonal $d = (v,v_a)$ is added. Note that $d$ is not crossed by an edge of $P_{in}$ as we have a Type 1 flip. Since $P_{in}$ is connected there has to exist a subpath of $P_{in}$ consisting of edges on the convex hull that connects $v$ with $v_a$. Thus, adding $d$ to $P_{in}$ results in a cycle that consists of $d$ and convex hull edges. Consequently, a convex hull edge is removed in the flip.	
		Therefore, $P_{tar}$ contains one diagonal more than $P_{in}$.
	\end{proof}
	
	\begin{figure}[ht]
		\centering
		\includegraphics[width=0.9\textwidth]{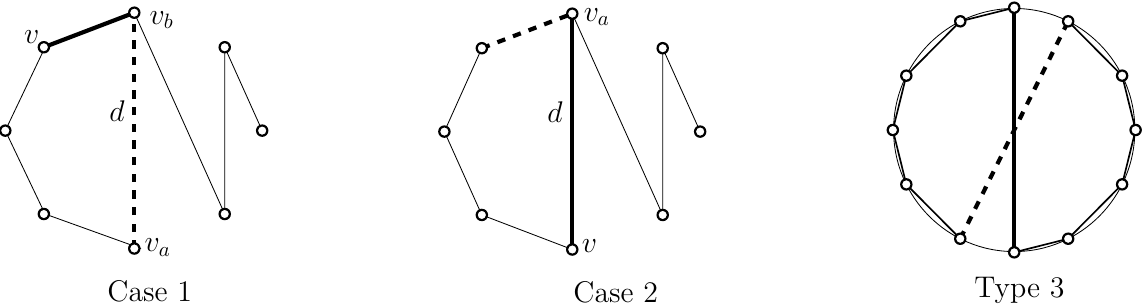}
		\caption{Case 1 (left): Adding a convex hull edge (fat), removing a diagonal (dashed). Case 2 (middle): Adding a diagonal (fat), removing a convex hull edge (dashed). For a Type 3 flip (right), we need all edges but one diagonal to be on the convex hull.}
		\label{Case12}
	\end{figure}
	
	The intuition behind Lemma~\ref{obs1} can be seen in Figure~\ref{Case12}. It shows all possible flips that can involve diagonals: Removing the unique diagonal $d$ that is visible from an end vertex~$v$ by adding a convex hull edge between $v$ and the vertex $v_b$ incident to~$d$ (left), adding a diagonal~$d$ by removing a convex hull edge (middle), and performing a Type 3 flip that exchanges the only diagonal with a diagonal that is rotated by one vertex (right). Observe that Type 3 flips need a special set-up. The diagonals that are involved in the flip need to be the only diagonals of their respective paths. Also, the convex hull edges need to coincide but emanate from opposite vertices.
	
	\begin{figure}[ht]
		\centering
		\includegraphics[width=0.9\textwidth]{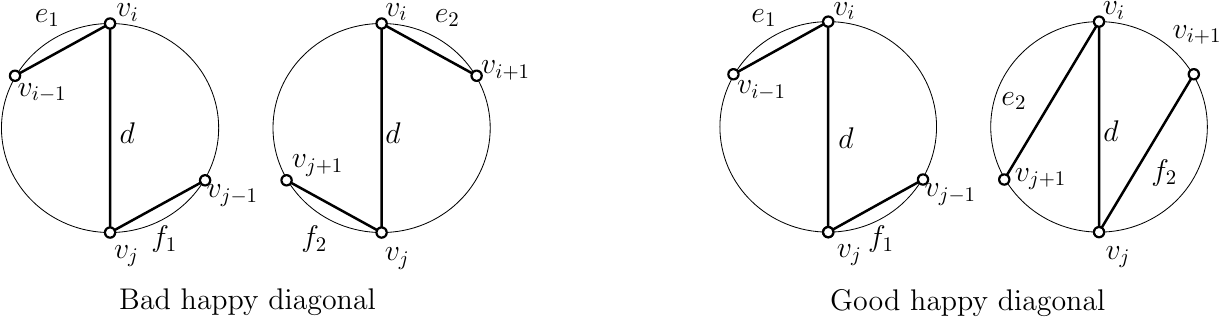}
		\caption[Good happy diagonals vs. bad happy diagonals.]{Bad happy diagonal $d$ with edges emanating to different sides in different paths (left), and good happy diagonal $d$ with edges emanating to the same side in different paths (right).} 
		\label{Good_bad}
	\end{figure}

	Let $d=(v_i,v_j)$ be a happy edge that is a diagonal. Note that $d$ splits the convex point set into two parts and that the two edges adjacent to $d$ emanate into one of the parts each. Let $d$ be adjacent to~$e_1$ at~$v_i$ and $f_1$ at $v_j$ in the initial path $P_{in}$, and to $e_2$ at $v_i$ and $f_2$ at $v_j$ in the target path~$P_{tar}$. We call $d$ a \emph{good happy diagonal} if $e_1$ and $e_2$ emanate to the same side of $d$, that is, the second vertices of $e_1$ and $e_2$ are either both in $\{v_{i+1},v_{j-1}\}$ or both in $\{v_{j+1},v_{i-1}\}$ where the vertices are considered cyclically. Otherwise, we call $d$ a \emph{bad happy diagonal} (if $e_1$ and $e_2$ emanate to different sides). See Figure~\ref{Good_bad} for an illustration.
	
	
	\begin{lemma} \label{obs2} Consider an initial path $P_{in}$ and a target path $P_{tar}$.
		\begin{itemize}
			\item[(a)] For every flip sequence from $P_{in}$ to $P_{tar}$  any bad happy diagonal needs to be removed.
			\item[(b)] For every subpath $R$ of consecutive happy edges that contains at least one good happy diagonal there exists a flip sequence from $P_{in}$ to $P_{tar}$ that preserves all edges of $R$.
		\end{itemize}
	\end{lemma}
	
	\begin{figure}
		\centering
		\includegraphics[width=0.9\textwidth]{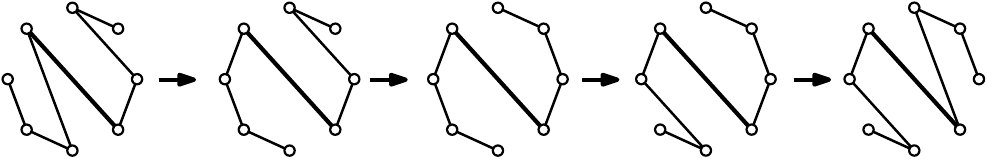}
		\caption{Illustration of Lemma \ref{obs2}(b): A flip sequence that preserves a good happy diagonal~$d$.}
		\label{good_happy}
	\end{figure}
	
	\begin{proof}
		Let $d$ be a bad happy diagonal, where $e_1$ and $f_1$ are its adjacent edges. For (a) observe that we cannot flip edges $e_1$ and $f_1$ in one step, so we need to flip $d$.
		
		For (b) start with $P_{in}$, pick one end vertex of $P_{in}$ and perform flips that remove diagonals until the next diagonal that is visible from the current endvertex lies in $R$. Then, repeat the process for the other endvertex of $P_{in}$. The resulting path $P'$ will then consist of $R$ and two (possibly empty) subpaths of convex hull edges. Do the same for $P_{tar}$ to obtain a path~$P''$. Since~$R$ contains a good happy diagonal and there are only two ways how the subpath of convex hull edges may look like, we get $P' = P''$. The required flip sequence from $P_{in}$ to~$P_{tar}$ consists of the flips from $P_{in}$ to $P'$ and the flips from $P_{tar}$ to~$P''$ in reverse order. For an example of such a flip sequence, see Figure \ref{good_happy}.
	\end{proof}
	
	We would like to point out the reason for studying diagonals in such detail. If we know all the diagonals of a path and the way the incident edges emanate from them, then we can already construct the entire path. As seen in Figure~\ref{intuition}, there is at most one way to complete the path between two consecutive diagonals. Also note that if a subpath of happy edges of length at least $2$ contains a diagonal, the diagonal has to be a good happy diagonal.
	These observations together simplify the problem of flipping an initial path into a target path to simply flipping all the diagonals with their emanating edges into the right place.
	
	\begin{figure}[ht]
		\centering
		\includegraphics[width=0.5\textwidth]{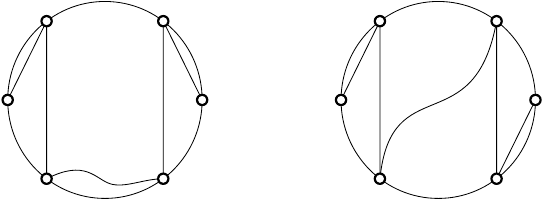}
		\caption{Diagonals and their emanating edges determine the whole path. The consecutive diagonals and emanating edges in the left path can be completed to a path, whereas the ones in the right path cannot be completed without adding another diagonal.}
		\label{intuition}
	\end{figure}
	
	\section{Characterization of Minimum Flip Sequences}
	
	Based on the structure of good happy diagonals and convex hull edges of the initial path and the target path, we provide a characterization of pairs of paths into four categories. We derive lower bounds on the number of flips and argue that for each category there exists a flip sequence that makes exactly this number of flips, thus providing a minimal flip sequence. Later, in Section \ref{linear}, we will show how to identify the four cases in linear time.
	
	\begin{theorem}\label{characterization}
		Let $P_{in}$ and $P_{tar}$ be two plane spanning paths for the same convex point set. Let $k$ and $l$ denote the number of diagonals of $P_{in}$ and $P_{tar}$, respectively. Then the flip distance between $P_{in}$ and $P_{tar}$ is described by the following cases.
		\begin{itemize}
			\item \textbf{Case 1:} If good happy diagonals exist, then let $m \geq 1$ be the maximal number of good happy diagonals in a subpath of consecutive happy edges. The flip distance is then $k+l-2m$.
			\item \textbf{Case 2:} If no good happy diagonal exists, then we distinguish three cases
			\subitem \textbf{Case 2a:} If there exists a pair of diagonals $d_1 \in P_{in}$, $d_2 \in P_{tar}$ that can eventually be exchanged for one another by a Type 3 flip, then the flip distance is $k+l-1$.
			\subitem \textbf{Case 2b:} If no Type 3 flip can be performed and if the diagonals can be flipped to convex hull edges in both, the initial and target path, such that the paths after flipping all diagonals coincide, the flip distance is $k+l$.
			\subitem \textbf{Case 2c:} Otherwise the flip distance is $k+l+1$.
		\end{itemize}
	\end{theorem}
	
	\begin{remark}\label{case2b}
		Regarding Case 2b: If a path $P$ does not contain the convex hull edge $(v_i,v_{i+1})$, we say that $P$ has a \emph{gap} $g$ in the convex hull at $(v_i,v_{i+1})$. If $g$ is a gap of both paths,~$P_{in}$ and $P_{tar}$, we say that~$P_{in}$ and $P_{tar}$ have a \emph{common gap} in the convex hull at $g$. Observe that by removing diagonals we can flip a path $P$ into a path that contains all convex hull edges except for the gap $g$ if and only if $g$ is a gap of $P$. This can be done by performing flips that add convex hull edges incident to each end vertex while removing diagonals until the next added convex hull edge would lie in $g$. For an intuitive example see Figure~\ref{Remark_4}.
	\end{remark}
	
	\begin{figure}[ht]
		\centering
		\includegraphics[width=0.9\textwidth]{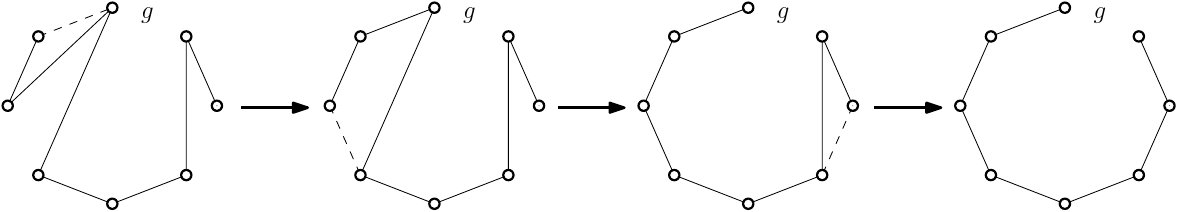}
		\caption{Illustration of Remark \ref{case2b}. In every step the convex hull edge to be added is indicated via a dashed line. First diagonals to the left of the gap $g$ are removed by the left end vertex of the path. Afterwards, the diagonal to the right gets removed by the right end vertex of the path.}
		\label{Remark_4}
	\end{figure}
	
	\begin{proof}[Proof of Theorem \ref{characterization}]
		For a visualization of the different cases, see Figure \ref{Vis}.
		
		\textbf{Case 1:} Consider a decomposition of the point set where points belong to the same component if and only if they lie on the same side of every good happy diagonal. Then, there are exactly two components that each contain one of the endpoints of the path. Therefore, to remove and add edges in components that do not contain endpoints, we need to remove all good happy diagonals on one side of this component. From this, we conclude that we need to remove the good happy diagonals from all but one subpath of happy edges. Let $m_R$ denote the number of good happy diagonals in a subpath of happy edges $R$.
		
		By Lemma \ref{obs1} and Lemma \ref{obs2}, removing all diagonals apart from the ones in $R$ takes at least~$k-m_R$ flips and can be done in that number of flips. Similarly, adding all the new diagonals takes $l-m_R$ flips. Therefore, there is a flip sequence from~$P_{in}$ to $P_{tar}$ that preserves~$R$ with~$k~+~l~-2m_R$ flips.~The length of the minimum flip sequence is therefore~$k+l-2m$ by the choice of $m$ and it is attained by applying the flips from the proof of Lemma \ref{obs2}. 
		
		
		\textbf{Case 2a:} The only way to exchange more than one diagonal at once is by performing a Type 3 flip. If we want to exchange one diagonal from $P_{in}$ directly with a diagonal in $P_{tar}$, all the flips leading up to the Type~3 flip need to remove the $k-1$ diagonals (that are not involved in the Type~3 flip) from the initial path. Similarly, all the flips that occur after the Type~3 flip add the $l-1$ diagonals of the target path that are not involved in the Type~3 flip. Further, the two subpaths of convex hull edges need to coincide between the paths before and after the Type~3 flip, see Figure \ref{Case12} (right). Therefore, the subpaths of convex hull edges are already correctly aligned after the Type 3 flip. This shows that $k+l-1$ flips are necessary and sufficient in case a Type 3 flip can occur: $k-1$ flips for removing diagonals from~$P_{in}$, one Type 3 flip and $l-1$ flips to add diagonals to get $P_{tar}$.
		
		\textbf{Case 2b:} If no Type 3 flip can be set up, it follows from Lemma~\ref{obs1} that $k$ flips are necessary and sufficient to remove all diagonals from the initial path, and similarly $l$ flips to add all the diagonals of the new path. Since the diagonals can be removed in a way that there is a common gap in the convex hull, $k+l$ flips are indeed necessary and sufficient.
		
		\textbf{Case 2c:} If neither Case 2a nor Case 2b hold, so in particular $P_{in}$ and $P_{tar}$ do not have a common gap in the convex hull, it is indeed necessary to remove all diagonals, realign the convex hull and add all diagonals to get the new path. So $k+l+1$ flips are necessary. 
		
		
		To show optimality also for Case 2c assume we could start adding diagonals from the target path before removing all diagonals from the initial path and realigning the path along the convex hull. We add the first diagonal $d=(v_i,v_j) \in P_{tar}$ in the step from $P_l$ to $P_{l+1}$. Assume $v_i$ was the end vertex in the previous step. By the assumption of Case 2c, $v_i$ is incident to two convex hull edges in~$P_{in} \cup P_{tar}$. One of them, say $e_1$, is in $ P_{tar}\setminus P_{in}$, otherwise there would not be an isolated vertex at $v_i$. Since $v_i$ has degree at most two in $P_{tar}$ the other convex hull edge, say $e_2$, has to be in $P_{in}\setminus P_{tar}$. Since we didn't remove any convex hull edge from $P_{in}$ prior to that flip, the edges incident to $v_i$, $e_1$ and $e_2$ emanate to different sides of~$d$ in~$P_{l+1}$ and~$P_{tar}$. Therefore, $d$ is a bad happy diagonal and has to be removed by Lemma \ref{obs2}a.
	\end{proof}
	
	\begin{figure}[ht]
		\centering
		\includegraphics[width=0.68\textwidth]{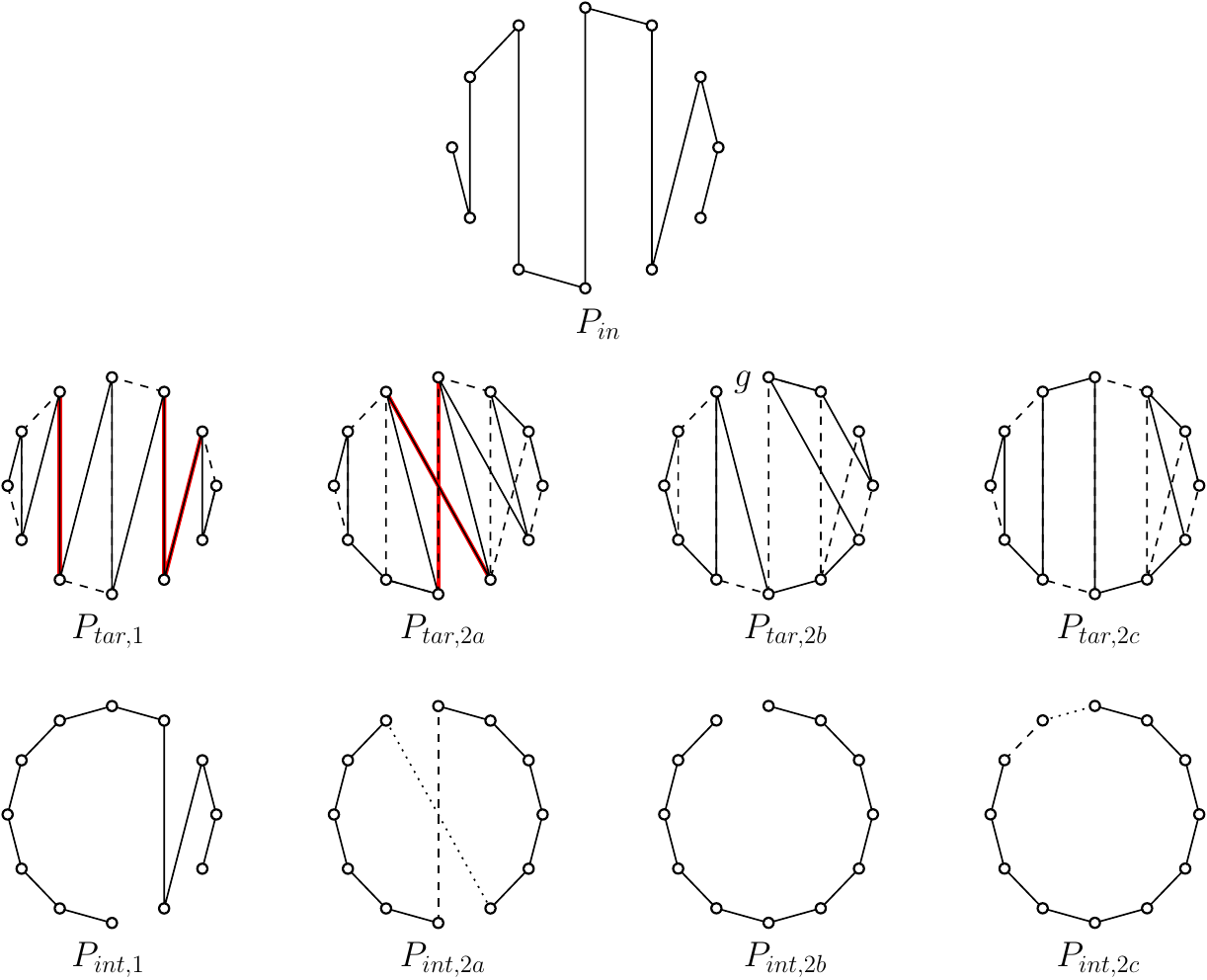}
		\caption{Visualization of Theorem \ref{characterization}: Top row: the initial path $P_{in}$. Second row: target paths where dashed lines hint at the initial path. The cases 1, 2a, 2b and 2c occur from left to right. In Case 1 good happy edges are marked in red. In Case 2a the two edges involved in a Type~3 flip are marked red. For Case 2b we labeled the common gap $g$ in the convex hull. Observe that for Case~2c, none of the other cases occur. Bottom row: intermediate configurations for the corresponding cases.}
		\label{Vis}
	\end{figure}
	
	
	\section{A Linear Time Approach \label{linear}}
	
	We show that for a given pair of paths the classification into the four categories derived in Theorem \ref{characterization} can be done in linear time.
	During this classification, we collect all the required information so that for each category we can then argue how the flips in a minimum flip sequence can be provided in linear time.
	
	\begin{theorem} \label{in_P}
		The flip distance between two plane spanning paths in convex point sets can be determined in $O(n)$ time and space.
	\end{theorem}
	
	\begin{proof}
		Recall that we label the vertices $v_0,\ldots,v_{n-1}$ in clockwise order, beginning at an arbitrary but fixed vertex. For both given paths, we create an array of size $n$, say $I[0,\ldots,n~-~1~]$ for the initial path, and $T[0,\ldots,n~-~1~]$ for the target path. For each array an entry at index $i$ contains the indices of the two (predecessor and successor, respectively) vertices the vertex $v_i$ is connected to in the respective path. For the end vertices of the two paths we only store one index each. These arrays can be initialized in linear time by traversing the two path and storing the indices. At the same time we compute the numbers $k$ and $l$ of diagonals in the initial and target path. This can simply be done by the observation that the vertices of diagonals do have non-consecutive indices, where the indices are taken modulo $n$.
		
		To see if \textbf{Case 1} occurs we check for the existence of good happy diagonals and also compute the number of good happy diagonals that appear in a consecutive sequence of happy edges. For this, traverse the edges of the initial path from one end to the other. For an edge $e=(v_i,v_j)$ where w.l.o.g.\ $i<j$ we check in constant time if it is happy by checking if $j$ is contained in the two indices of $T[i]$. If this is the case we have a happy diagonal if $i$ and $j$ are not consecutive, taken modulo $n$. For a happy diagonal $(v_i,v_j)$ we then check if it is good. Let $v_k$ and $v_{k'}$ be the successors of $v_j$ in the initial and target path, respectively. The indices~$k$ and $k'$ can be obtained in constant time from $I[j]$ and $T[j]$. If $i<k<j$ and $i<k'<j$ is either both true or both false we have a good happy diagonal, otherwise a bad happy diagonal. While traversing all edges we also maintain a counter $m$, a start index, and an end index to keep track of the subpath $R$ of consecutive happy edges with the highest number of good happy diagonals. If at the end of the process we have $m>0$, then the minimum flip sequence has length $l+k-2m$. The flip sequence can be obtained by first removing all diagonals from the initial path except the ones in $R$ and then adding all missing diagonals from the target path.
		
		However, if $m = 0$, then we need to check for further cases. We continue by checking, whether the diagonals can be removed in such a way, that they permit a Type 3 flip, that is, whether we have \textbf{Case 2a}.
		For this, again traverse the edges of the initial path. For every diagonal $(v_i,v_j)$ where w.l.o.g. $i<j$ check the following: If $v_i$ is adjacent to a vertex $v_k$ with $i<k<j$ in the initial path, then there has to be the edge $(v_{i-1},v_{j+1})$ (indices taken modulo~$n$) in the target path and $v_{i-1}$ has to be adjacent to a vertex $v_{k'}$ with $j+1 < k'$ or $k'<i-1$. Checking whether $(v_i,v_j)$ is a diagonal can be done in constant time for every pair of indices $i$ and $j$ by checking if the indices are non-consecutive. Checking the existence of the shifted diagonal in the target path can also be done in constant time by accessing the entry~$T[i+1]$ and checking whether $v_{j-1}$ appears in the target path as one of the two neighbors of~$v_{i+1}$. The index $k'$ happens to be the entry in $T[i+1]$ that differs from $j-1$. We conclude that all the checks for a given diagonal can be executed in constant time. Conversely, if $v_i$ is adjacent to a vertex $v_k$ with $j<k$ or $k<i$ in the initial path, then there has to be the edge $(v_{i+1},v_{j-1})$ (indices taken modulo n) in the target path and~$v_{i+1}$ has to be adjacent to a vertex $v_{k'}$ with $i+1 < k'<j-1$. As before, all checks can be done in constant time for a fixed given diagonal. If a pair of edges $(v_i,v_j)$ and $(v_{i+1},v_{j-1})$ ($v_{i-1},v_{j+1}$ respectively) that permits a Type 3 flip exists, remove all $k-1$ diagonals except for the diagonal $(v_i,v_j)$, then perform a Type 3 flip to flip in $(v_{i+1},v_{j-1})$ or $(v_{i-1},v_{j+1})$, respectively. Finally continue by adding all the remaining $l-1$ diagonals to obtain a flip sequence of length $k+l-1$,
		
		If no pair of edges for a Type 3 flip exists, we next check, whether the edges can be removed in a way such that they permit a common gap on the convex hull after removing all diagonals, that is, whether \textbf{Case 2b} occurs. For this, create an additional array~$C[0,\ldots,~n~-~1~]$ where the entry $i$ represents the convex hull edge $(v_i,v_{i+1})$ (entries taken modulo $n$).  Each entry will be either marked or not. In the beginning, every entry is not marked. We first traverse the initial path and mark all entries of the array where the corresponding convex hull edge is part of the initial path. This can again be done by checking whether consecutive vertices along the path are described by consecutive numbers modulo $n$ and then marking the entry of~$C$ that corresponds to the lower of the two indices (with the exception of the edge $(v_{n-1},v_0)$ where we mark $C[n-1]$). We repeat the process for the target path. Afterwards, we traverse the array $C$ and check, if we can find an unmarked entry. An unmarked entry corresponds to a convex hull edge $e$ that does not lie in either, the initial path and the target path. By Remark \ref{case2b}, $e$ does not have to be flipped in, when removing and adding the diagonals in the right order. This can be done by removing diagonals on both ends of a path until the next edge to be flipped in would be $e$. After all~$k$ diagonals are removed, we start flipping in the $l$ diagonals to obtain the target path. The constructed flip sequence has length $k+l$.
		
		If none of the cases above occur the only remaining case is \textbf{Case 2c}. For a minimum flip sequence just flip all diagonals to the convex hull from either side in both the initial and target path to obtain two paths entirely on the convex hull. Then, the final flip sequence consists of going from the initial path to the first path on the convex hull. Afterwards, we perform a Type 2 flip to obtain the other path on the convex hull and then we reverse the flip sequence from the target path to said convex hull path, to get to the target path. The obtained flip sequence has length $k+l+1$.
	\end{proof}
	
	We note that to check every case we only traverse every path once and for every edge we perform checks that can be done in constant time. This yields a linear time algorithm for finding the minimum flip sequence between pairs of plane spanning paths in convex point sets. In our description, the paths are traversed up to four times, once for initialization and once for each case that checks the existence of certain structures. It would be possible to perform the latter three checks within one iteration, but since it doesn't affect the asymptotic runtime we opted for the description that is easier to follow.

	\section{Variants of the Problem}
	
	In the following, we describe variants of our initial problem of finding the flip distance between plane spanning paths in convex point sets. 
	
	First, we consider a variation of the problem in which we only allow flips for which the removed edge does not cross the added edge, that is, the union of the two paths that are part of the flip is crossing free. With our notion of flips, this means that we restrict the set of flips to flips of Type 1 and Type 2. We call the length of a minimum flip sequence in this setting the \emph{compatible flip distance}. Restricting existing flips to the case where the initial and resulting graphs are compatible has also been of interest for plane spanning trees~\cite{NICHOLS2020111929} or crossing-free perfect matchings~\cite{articlematchings_2}.
	
	Afterwards, we consider a local variant of the problem where we require the removed edge and the added edge to share a vertex. In this setting Type 1 flips still work as before. Type~2 flips, however, can only be simulated by repeatedly adding an edge and removing an edge that follows immediately afterwards. This can be interpreted as performing Type~2 flips at a higher cost, which is the distance between the added and the removed edge in the path. In this setting we call a valid flip a \emph{local flip} and the length of a minimum flip sequence the \emph{local flip distance}. The setting, where the added and removed edge have to share a vertex is known as edge rotation in the setting of plane spanning trees~\cite{NICHOLS2020111929}.
	
	As a last variant we consider paths in simple polygons. A path in a simple polygon can have two types of edges. The first type are edges that are boundary edges of the polygon. The second type lie in the interior of the polygon.
	
	For all three variants we describe how to modify our algorithm to obtain the flip distance in linear time. Especially for paths in simple polygons this might be surprising as we recall that the flip distance problem for triangulations of simple polygons is NP-complete \cite{Aichholzer2015}. 
	
	\subsection{Compatible Flip Distance in Convex Point Sets}
	
	We remark that for the compatible version of finding minimum flip sequences between plane spanning paths in convex point sets the only difference to the original setting is that we do not allow Type 3 flips. Therefore, the proof of Theorem \ref{compatibleflips} below is the same as for Theorems \ref{characterization} and \ref{in_P} except that we skip the original Case 2a. The notion of good and bad happy diagonals stay the same as for the original setting.
	
	\begin{theorem}\label{compatibleflips}
		Let $P_{in}$ and $P_{tar}$ be two plane spanning paths for the same convex point set. Let $k$ and $l$ denote the number of diagonals of $P_{in}$ and $P_{tar}$, respectively. Then the compatible flip distance between $P_{in}$ and $P_{tar}$ is described by the following cases.
		\begin{itemize}
			\item \textbf{Case 1:} If good happy diagonals exist, let $m \geq 1$ be the maximal number of good happy diagonals in a subpath of consecutive happy edges. The compatible flip distance is then $k+l-2m$.
			\item \textbf{Case 2:} If no good happy diagonal exists, then we distinguish two cases
			\subitem \textbf{Case 2a:} If the diagonals can be flipped to convex hull edges in both, the initial and target path, such that the paths after flipping all diagonals coincide, the compatible flip distance is $k+l$.
			\subitem \textbf{Case 2b:} Otherwise the compatible flip distance is $k+l+1$.
		\end{itemize}
		All the checks for the case distinction above can be implemented to run in total time and memory that is linear in the instance size.
	\end{theorem}
	
	\subsection{Local Flip Distance in Convex Point Sets}
	
	If we restrict the flips to the local variant where the removed an added edge have to share a vertex, we observe again that Type 3 flips can no longer occur. Further, re-aligning the convex hull can no longer be done in a single flip if the two gaps are too far apart. We partially resolve the second issue with the following lemma:
	
	\begin{lemma}\label{type2local}
		Any Type 2 flip between two paths in convex point sets can be simulated using at most $2$ local flips.
	\end{lemma}
	
	\begin{figure}[ht]
		\centering
		\includegraphics[width=0.9\textwidth]{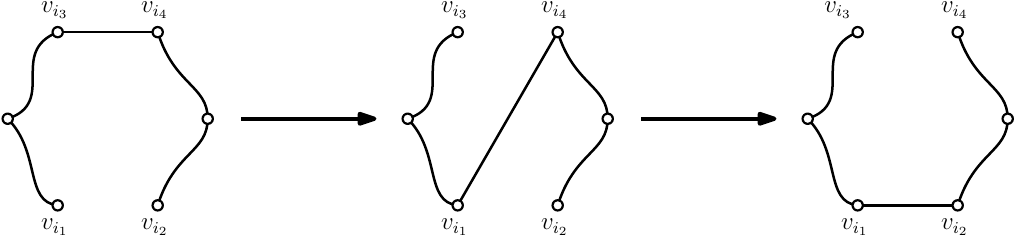}
		\caption{Replacing a Type 2 flip with two local flips}
		\label{Type2local}
	\end{figure}
	
	\begin{proof}
		Let $P_{in} \neq P_{tar}$ be two paths on a convex point set such that all their edges lie on the convex hull. Let $g_1$ and $g_2$ be their respective gaps with $g_1\neq g_2$.
		
		If $g_1$ and $g_2$ share a vertex, then the flip from $P_{in}$ to $P_{tar}$ is a local flip.
		
		Else, if $g_1$ and $g_2$ do not share a vertex, let $v_{i_1}$ and $v_{i_2}$ be the vertices of $g_1$. Further, let $v_{i_3}$ and $v_{i_4}$ be the vertices of $g_2$ such that $v_{i_3}$ occurs first when walking along $P_{in}$ in the direction from $v_{i_1}$ to $v_{i_2}$.
		
		We can flip from $P_{in}$ to $P_{tar}$ via the following two flips: (1) Add $(v_{i_1},v_{i_4})$, remove $(v_{i_3},v_{i_4})$ and (2) add $(v_{i_1},v_{i_2})$, remove $(v_{i_1},v_{i_4})$.
		See Figure \ref{Type2local} for an illustration.
	\end{proof}
	
	We conclude that the proof of Theorem \ref{localflips} below works similar to for Theorem \ref{characterization} and~\ref{in_P} for most of the cases. Again, the original Case 2a has no counterpart in this setting as Type~3 flips are prohibited. Further, once we are forced to perform flips that realign the convex hull, we need to distinguish further, whether this happens in one or two flips. The checks that are needed for this case distinction can still be done in linear time and space. When checking for mutual gaps in the convex hull we not only check whether the two paths have a common gap, but also whether they have adjacent gaps. We revisit our array $C[0,\ldots,n-1]$, where entry $i$ represents the edge~$(v_i,v_{i+1})$ (taken modulo $n$). Now the entries of $C$ are no longer marked or unmarked, but $C$ will contain integers from $0$ to~$3$. In the beginning, we set every value of $C$ to $0$. We traverse the initial path. Every time we encounter a convex hull edge~$(v_i,v_{i+1})$ along the path, we increase entry $i$ in $C$ by $1$. Then we traverse the target path and increase entry $i$ by $2$, if we encounter the convex hull edge~$(v_i,v_{i+1})$. Finally, we traverse $C$ again. An entry $0$ of $C$ will correspond to a common gap in the convex hull. Values $1$ and $2$ that appear consecutively in some order correspond to gaps that share a single vertex. We traverse the paths and the array a constant number of times, performing constant time checks and updates along the way. This yields a linear time algorithm to check all the conditions for the characterization in Theorem \ref{localflips}.
	
	We derive the following characterization of minimum flip sequences:
	
	\begin{theorem}\label{localflips}
		Let $P_{in}$ and $P_{tar}$ be two plane spanning paths for the same convex point set. Let $k$ and $l$ denote the number of diagonals of $P_{in}$ and $P_{tar}$, respectively. Then the local flip distance between $P_{in}$ and $P_{tar}$ is described by the following cases.
		\begin{itemize}
			\item \textbf{Case 1:} If good happy diagonals exist, let $m \geq 1$ be the maximal number of good happy diagonals in a subpath of consecutive happy edges. The local flip distance is then $k+l-2m$.
			\item \textbf{Case 2:} If no good happy diagonal exists, then we distinguish three cases
			\subitem \textbf{Case 2a:} If the diagonals can be flipped to convex hull edges in both, the initial and target path, such that the paths after flipping all diagonals coincide, the local flip distance is $k+l$.
			\subitem \textbf{Case 2b:} Otherwise if the diagonals can be removed such that the gaps in the convex hull share a vertex, then the local flip distance is $k+l+1$.
			\subitem \textbf{Case 2c:} Otherwise, the local flip distance is $k+l+2$.
		\end{itemize}
		All the checks for the case distinction above can be implemented to run in total time and memory that is linear in the instance size.
	\end{theorem}
	
	While for specific pairs of paths the flip distance may increase by $1$ when considering the local flip distance, this does not have any impact on the diameter of the flip graph. The reason is that pairs of paths for which the flip distance is maximized need to have many diagonals and consequently many gaps in the convex hull. As soon as $k+l > n$, the paths share a gap in the convex hull by the pidgeonhole principle and thus Cases 2b and 2c of Theorem~\ref{localflips} cannot occur if sufficiently many diagonals are present.
	
	\subsection{Flip Distance in Simple Polygons}
	
	For a given simple polygon, we consider plane spanning paths in the visibility graph of the polygon. The vertex set of the path is the set of vertices of the polygon and the set of edges is described by (1) the edges on the boundary of the polygon, in the following called \emph{boundary edges} and (2) \emph{interior edges} between two non-consecutive vertices of the polygon such that these straight line connection between the two vertices is entirely contained in the interior of the polygon. We observe that interior edges cut the polygon into two parts such that the two parts cannot interact as long as the interior edge exists.
	
	In analogy to the preliminaries in Section \ref{basics} we observe the following behavior of the different types of flips. See Figure \ref{flips_polygons} for an illustration.
	
	\begin{figure}[ht]
		\centering
		\includegraphics[width=0.9\textwidth]{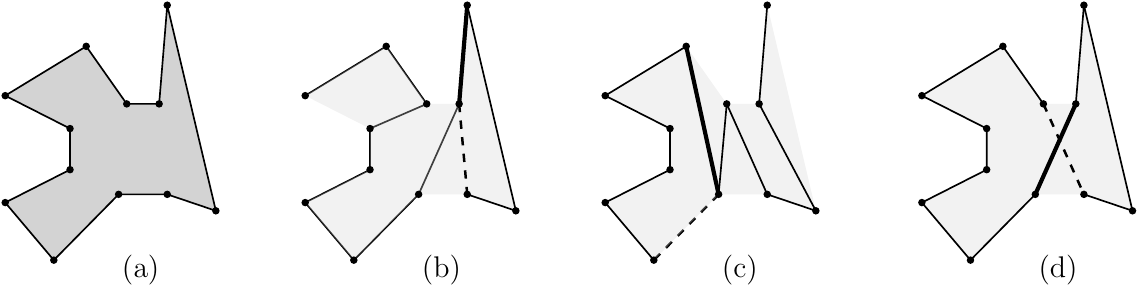}
		\caption{Flips between paths in simple polygons that change interior edges:
			\begin{itemize}
				\item[(a)] The underlying simple polygon
				\item[(b)] A Type 1 flip that adds a boundary edge (fat) and removes an interior edge (dashed)
				\item[(c)] A Type 1 flip that adds an interior edge (fat) and removes a boundary edge (dashed)
				\item[(d)] A Type 3 flip that exchanges two interior edges (fat and dashed)
		\end{itemize}}
		\label{flips_polygons}
	\end{figure}
	
	\begin{itemize}
		\item A Type 1 flip either adds a boundary edge and removes and interior edge or adds an interior edge and removes a boundary edge. Further, if there exist interior edges, we can always perform a Type 1 flip that reduces the number of interior edges.
		\item A Type 2 flip can be performed if and only if all the edges of the path are boundary edges. A Type 2 flip exchanges two paths that are entirely contained in the boundary for one another.
		\item For a Type 3 flip to be possible, for each path there has to be exactly one interior edge and all the other edges have to be on the boundary. In the setting of simple polygons we additionally require that the straight line that connects the two endpoints of a path lies entirely inside the polygon.
	\end{itemize}
	
	Next, we take a look at happy interior edges, that is, interior edges that lie in both the initial path and the target path. We observe the following analogy to Lemma \ref{obs2}. See Figure~\ref{good_bad_polygons} for an illustration.
	
	\begin{itemize}
		\item \emph{Good happy interior edges} are interior edges that split the polygon into two parts such that the subpaths emanating from one vertex of the interior edge emanate into the same part of the polygon in both the initial path and the target path. There exist flip sequences from the initial path to the target paths such that one sequence of consecutive happy edges containing good happy interior edges is not flipped.
		\item \emph{Bad happy interior edges} are interior edges for which their emanating subpaths at each vertex emanate to different sides of the polygon. Every flip sequence from the initial path to the target path has to remove all bad happy interior edges.
	\end{itemize}
	
	\begin{figure}[ht]
		\centering
		\includegraphics[width=0.9\textwidth]{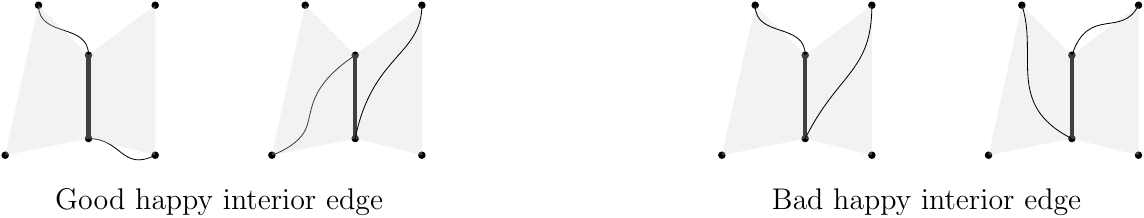}
		\caption{Left: Good happy interior edge with subpaths emanating to the same side\\
			Right: Bad happy interior edge with subpaths emanating to different sides}
		\label{good_bad_polygons}
	\end{figure}
	
	Finally, we observe as before that the information about interior edges and the way how subpaths emanate from them is sufficient to reconstruct the entire path. This shifts the problem of flipping paths to the problem of positioning diagonals correctly. We have obtained all the necessary ingredients so that we can follow the lines of the proof of Theorem~\ref{characterization} and can thus conclude:
	
	\begin{theorem}\label{polygons}
		Let $P_{in}$ and $P_{tar}$ be two plane spanning paths in the same simple polygon. Let $k$ and $l$ denote the number of interior edges of $P_{in}$ and $P_{tar}$, respectively. Then the flip distance between $P_{in}$ and $P_{tar}$ is described by the following cases.
		\begin{itemize}
			\item \textbf{Case 1:} If good happy interior edges exist, let $m \geq 1$ be the maximal number of good happy interior edges in a subpath of consecutive happy edges. Then the flip distance is $k+l-2m$.
			\item \textbf{Case 2:} If no good happy interior edges exists, then we distinguish three cases
			\subitem \textbf{Case 2a:} If there exists a pair of interior edges $d_1 \in P_{in}$, $d_2 \in P_{tar}$ that can eventually be exchanged for one another by a Type 3 flip, then the flip distance is $k+l-1$.
			\subitem \textbf{Case 2b:} If no Type 3 flip can be performed and if the interior edges can be flipped to boundary edges in both, the initial and target path, such that the paths after flipping all interior edges coincide, the flip distance is $k+l$.
			\subitem \textbf{Case 2c:} Otherwise the flip distance is $k+l+1$.
		\end{itemize}
		All the checks for the case distinction above can be implemented to run in total time and memory that is linear in the instance size.
	\end{theorem}
	
	From Theorem \ref{polygons} we can bound the number of flips that it takes to go from one path in a simple polygon to any other.
	
	\begin{corollary}\label{cor:diam}
		The diameter of the flip graph of paths in a simple polygon of size $n$ is $\leq 2n-6$ for~$n>4$ and $\leq 2n-5$ for $n\in\{3,4\}$.
	\end{corollary}
	
	\begin{proof}
		Any path in a simple polygon has at most $n-3$ interior edges. With the notation of Theorem \ref{polygons}, $k+l \leq 2n-6$. Additionally, we observe that for every interior edge, we lose one boundary edge and get an additional gap in the boundary. So, if $k+l \geq n-1$, there is a total of $\geq n+1$ gaps in $P_{in}$ and $P_{tar}$. Then, by the pidgeonhole principle, the two paths have to share a gap in the boundary, and Case 2c of Theorem \ref{polygons} cannot occur. If the polygon has $3$ or $4$ vertices there could exist two paths that together cover the entire boundary of the polygon. For $n>4$ the $2n-6\geq n-1$ and for any choice of $k$ and $l$ we get a flip distance of at most $2n-6$.
	\end{proof}
	
	The upper bound of Corollary \ref{cor:diam} is attained for convex polygons. In general, the bound is not tight. The interested reader may check that the polygon in Figure \ref{icecream} does not admit paths with more than one diagonal. Therefore, by Theorem \ref{polygons}, any path in this polygon can be flipped into any other path in at most $3$ flips. Further, the construction of the polygon can be generalized to contain an arbitrary number of vertices.
	
	\begin{figure}[ht]
		\centering
		\includegraphics[width=0.7\textwidth]{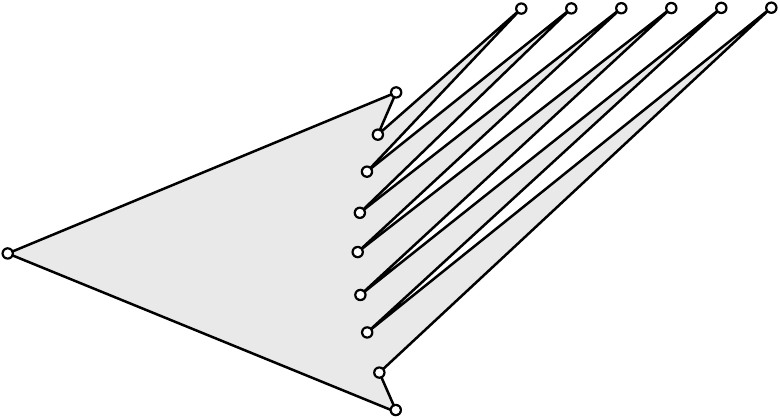}
		\caption{Polygon that only allows paths with at most one interior edges and where the flip graph has diameter $3$}
		\label{icecream}
	\end{figure}
	
	
	
	\break
	
	\section{Conclusion}
	
	We disproved the happy edge property for plane spanning paths on convex point sets. At the same time, we provided a linear time algorithm to compute the flip distance between two given plane spanning paths. This contradicts the assumption that the absence of the happy edge property makes finding minimum flip sequences hard to solve. We are not aware of any problem for which the opposite direction fails, that is, a reconfiguration problem for which the happy edge property is true, but it is hard to compute the flip distance.
	
	The central questions about the flip graph of paths in point sets in general position still remain open: Is the flip graph connected? And if yes, how many flips do we need in the worst case? A related question is, how hard it is to find a shortest flip sequence between two spanning paths (or to decide that there exists none)?
	
	We have also seen that we can determine the flip distance between paths in simple polygons in linear time. A natural question to generalize this would be to consider questions about paths in polygons with holes.
	
	
	%
	
	\bibliography{paths_arXiv.bib}
	
\end{document}